\documentclass[11pt,oneside,reqno,american]{amsart}
\usepackage[T1]{fontenc}
\usepackage[latin9]{inputenc}
\usepackage{verbatim}
\usepackage{mathrsfs}
\usepackage{amstext}
\usepackage{amsthm}
\usepackage{amssymb}
\usepackage{undertilde}
\usepackage{xargs}[2008/03/08]
\usepackage[all]{xy}

\makeatletter
\numberwithin{equation}{section}
\numberwithin{figure}{section}
\theoremstyle{plain}
\newtheorem{thm}{\protect\theoremname}[section]
\theoremstyle{definition}
\newtheorem{defn}[thm]{\protect\definitionname}
\theoremstyle{remark}
\newtheorem{rem}[thm]{\protect\remarkname}
\theoremstyle{plain}
\newtheorem{lem}[thm]{\protect\lemmaname}

\usepackage{mathtools}
\usepackage{amsthm}
\usepackage{pdfsync} 
\usepackage{hyperref}
\usepackage[all]{xy}
\newcommand{\ie}{\textit{i.e.}}
\newcommand{\eg}{\textit{e.g.}}
 \theoremstyle{plain}
  \newtheorem{myprop}[thm]{Proposition}
\usepackage{fbb}
\usepackage[libertine,bigdelims]{newtxmath}
\usepackage[cal=boondoxo,bb=boondox,frak=boondox]{mathalfa}
\usepackage{xcolor} 
\definecolor{brown(traditional)}{rgb}{0.59, 0.29, 0.0}
\definecolor{blue(ryb)}{rgb}{0.01, 0.28, 1.0}
\definecolor{red}{rgb}{1.0, 0.0, 0.0}
\definecolor{magenta}{rgb}{1.0, 0.0, 1.0}
\definecolor{mahogany}{rgb}{0.75, 0.25, 0.0}
\definecolor{lavenderpurple}{rgb}{0.59, 0.48, 0.71}
\definecolor{olive}{rgb}{0.5, 0.5, 0.0}
\definecolor{brickred}{rgb}{0.8, 0.25, 0.33}
\definecolor{antiquefuchsia}{rgb}{0.57, 0.36, 0.51}
\definecolor{bole}{rgb}{0.47, 0.27, 0.23}
\definecolor{darkolivegreen}{rgb}{0.33, 0.42, 0.18}
\definecolor{deepjunglegreen}{rgb}{0.0, 0.29, 0.29}
\definecolor{brickred}{rgb}{0.8, 0.25, 0.33}
\definecolor{deepjunglegreen}{rgb}{0.0, 0.29, 0.29}
\definecolor{darkpastelgreen}{rgb}{0.01, 0.75, 0.24}
\definecolor{green(pigment)}{rgb}{0.0, 0.65, 0.31}
\definecolor{junglegreen}{rgb}{0.16, 0.67, 0.53}
\definecolor{officegreen}{rgb}{0.0, 0.5, 0.0}
\definecolor{seagreen}{rgb}{0.18, 0.55, 0.34}
\definecolor{teal}{rgb}{0.0, 0.5, 0.5}
\definecolor{brightgreen}{rgb}{0.4, 1.0, 0.0}
\definecolor{electricgreen}{rgb}{0.0, 1.0, 0.0}
\definecolor{malachite}{rgb}{0.04, 0.85, 0.32}

\makeatother

\usepackage{babel}
\providecommand{\definitionname}{Definition}
\providecommand{\lemmaname}{Lemma}
\providecommand{\remarkname}{Remark}
\providecommand{\theoremname}{Theorem}

\begin{document}

\global\long\def\ga{\alpha}%
\global\long\def\gb{\beta}%
\global\long\def\ggm{\gamma}%
\global\long\def\go{\omega}%
\global\long\def\gs{\sigma}%
\global\long\def\gd{\delta}%
\global\long\def\gD{\Delta}%
\global\long\def\vph{\phi}%
\global\long\def\gf{\varphi}%
\global\long\def\gk{\kappa}%
\global\long\def\gl{\lambda}%
\global\long\def\gz{\zeta}%
\global\long\def\gh{\eta}%
\global\long\def\gy{\upsilon}%
\global\long\def\gth{\theta}%
\global\long\def\gO{\Omega}%
\global\long\def\gG{\Gamma}%

\global\long\def\eps{\varepsilon}%
\global\long\def\epss#1#2{\varepsilon_{#2}^{#1}}%
\global\long\def\ep#1{\eps_{#1}}%

\global\long\def\wh#1{\widehat{#1}}%
\global\long\def\hi{\hat{\imath}}%
\global\long\def\hj{\hat{\jmath}}%
\global\long\def\hk{\hat{k}}%
\global\long\def\ol#1{\overline{#1}}%
\global\long\def\ul#1{\underline{#1}}%

\global\long\def\spec#1{\textsf{#1}}%

\global\long\def\ui{\wh{\boldsymbol{\imath}}}%
\global\long\def\uj{\wh{\boldsymbol{\jmath}}}%
\global\long\def\uk{\widehat{\boldsymbol{k}}}%

\global\long\def\uI{\widehat{\mathbf{I}}}%
\global\long\def\uJ{\widehat{\mathbf{J}}}%
\global\long\def\uK{\widehat{\mathbf{K}}}%

\global\long\def\bs#1{\boldsymbol{#1}}%
\global\long\def\vect#1{\mathbf{#1}}%
\global\long\def\bi#1{\textbf{\emph{#1}}}%

\global\long\def\uv#1{\widehat{\boldsymbol{#1}}}%
\global\long\def\cross{\times}%

\global\long\def\ddt{\frac{\dee}{\dee t}}%
\global\long\def\dbyd#1{\frac{\dee}{\dee#1}}%
\global\long\def\dby#1#2{\frac{\partial#1}{\partial#2}}%
\global\long\def\dxdt#1{\frac{\dee#1}{\dee t}}%

\global\long\def\vct#1{\bs{#1}}%

\global\long\def\partialby#1#2{\frac{\partial#1}{\partial x^{#2}}}%
\newcommandx\parder[2][usedefault, addprefix=\global, 1=]{\frac{\partial#2}{\partial#1}}%

\global\long\def\fall{,\quad\text{for all}\quad}%

\global\long\def\reals{\mathbb{R}}%

\global\long\def\rthree{\reals^{3}}%
\global\long\def\rsix{\reals^{6}}%
\global\long\def\rn{\reals^{n}}%

\global\long\def\prn{\reals^{n+}}%
\global\long\def\nrn{\reals^{n-}}%
\global\long\def\cprn{\overline{\reals}^{n+}}%
\global\long\def\cnrn{\overline{\reals}^{n-}}%
\global\long\def\rt#1{\reals^{#1}}%
\global\long\def\rtw{\reals^{12}}%

\global\long\def\les{\leqslant}%
\global\long\def\ges{\geqslant}%

\global\long\def\dee{\textrm{d}}%
\global\long\def\di{d}%
\global\long\def\dX{\dee\bp}%
\global\long\def\dx{\dee x}%
\global\long\def\D{D}%

\global\long\def\from{\colon}%
\global\long\def\tto{\longrightarrow}%
\global\long\def\lmt{\longmapsto}%
\global\long\def\lhr{\lhook\joinrel\longrightarrow}%
\global\long\def\mto{\mapsto}%

\global\long\def\abs#1{\left|#1\right|}%

\global\long\def\isom{\cong}%

\global\long\def\comp{\circ}%

\global\long\def\cl#1{\overline{#1}}%

\global\long\def\fun{\varphi}%

\global\long\def\interior{\textrm{Int}\,}%
\global\long\def\inter#1{\kern0pt  #1^{\mathrm{o}}}%
\global\long\def\interior{\textrm{Int}\,}%
\global\long\def\inter#1{\kern0pt  #1^{\mathrm{o}}}%
\global\long\def\into{\mathrm{o}}%

\global\long\def\sign{\textrm{sign}\,}%
\global\long\def\sgn#1{(-1)^{#1}}%
\global\long\def\sgnp#1{(-1)^{\abs{#1}}}%

\global\long\def\du#1{#1^{*}}%

\global\long\def\tsum{{\textstyle \sum}}%
\global\long\def\lsum{{\textstyle \sum}}%

\global\long\def\dimension{\textrm{dim}\,}%

\global\long\def\esssup{\textrm{ess}\,\sup}%

\global\long\def\ess{\textrm{{ess}}}%

\global\long\def\kernel{\mathop{\textrm{\textup{Kernel}}}}%

\global\long\def\support{\mathop{\textrm{\textup{supp}}}}%

\global\long\def\image{\mathop{\textrm{\textup{Image}}}}%

\global\long\def\diver{\mathop{\textrm{\textup{div}}}}%

\global\long\def\spanv{\textrm{span}}%

\global\long\def\tr{\mathop{\textrm{\textup{tr}}}}%
\global\long\def\tran{\mathrm{tr}}%

\global\long\def\opt{\mathrm{opt}}%

\global\long\def\resto#1{|_{#1}}%
\global\long\def\incl{\mathcal{I}}%
\global\long\def\iden{\imath}%
\global\long\def\idnt{\textrm{Id}}%
\global\long\def\rest{\rho}%
\global\long\def\extnd{e_{0}}%

\global\long\def\proj{\textrm{pr}}%

\global\long\def\L#1{L\bigl(#1\bigr)}%
\global\long\def\LS#1{L_{S}\bigl(#1\bigr)}%

\global\long\def\ino#1{\int_{#1}}%

\global\long\def\half{\frac{1}{2}}%
\global\long\def\shalf{{\scriptstyle \half}}%
\global\long\def\third{\frac{1}{3}}%

\global\long\def\empt{\varnothing}%

\global\long\def\paren#1{\left(#1\right)}%
\global\long\def\bigp#1{\bigl(#1\bigr)}%
\global\long\def\biggp#1{\biggl(#1\biggr)}%
\global\long\def\Bigp#1{\Bigl(#1\Bigr)}%

\global\long\def\braces#1{\left\{  #1\right\}  }%
\global\long\def\sqbr#1{\left[#1\right]}%
\global\long\def\anglep#1{\left\langle #1\right\rangle }%

\global\long\def\bigabs#1{\bigl|#1\bigr|}%
\global\long\def\dotp#1{#1^{\centerdot}}%
\global\long\def\pdot#1{#1^{\bs{\!\cdot}}}%

\global\long\def\eq{\sim}%
\global\long\def\quot{/\!\!\eq}%
\global\long\def\by{\!/\!}%

\global\long\def\stp{\text{\small\ensuremath{\bigodot}}}%
\global\long\def\tp{\text{\small\ensuremath{\bigotimes}}}%

\global\long\def\mi#1{#1}%
\global\long\def\mii{I}%
\global\long\def\mie#1#2{#1_{1}\cdots#1_{#2}}%

\global\long\def\smi#1{\boldsymbol{#1}}%
\global\long\def\asmi#1{#1}%
\global\long\def\ordr#1{\left\langle #1\right\rangle }%

\global\long\def\symm#1{\paren{#1}}%
\global\long\def\smtr{\mathcal{S}}%

\global\long\def\perm{p}%
\global\long\def\sperm{\mathcal{P}}%

\global\long\def\oneto{1,\dots,}%

\global\long\def\lisub#1#2#3{#1_{1}#2\dots#2#1_{#3}}%

\global\long\def\lisup#1#2#3{#1^{1}#2\dots#2#1^{#3}}%

\global\long\def\lisubb#1#2#3#4{#1_{#2}#3\dots#3#1_{#4}}%

\global\long\def\lisubbc#1#2#3#4{#1_{#2}#3\cdots#3#1_{#4}}%

\global\long\def\lisubbwout#1#2#3#4#5{#1_{#2}#3\dots#3\widehat{#1}_{#5}#3\dots#3#1_{#4}}%

\global\long\def\lisubc#1#2#3{#1_{1}#2\cdots#2#1_{#3}}%

\global\long\def\lisupc#1#2#3{#1^{1}#2\cdots#2#1^{#3}}%

\global\long\def\lisupp#1#2#3#4{#1^{#2}#3\dots#3#1^{#4}}%

\global\long\def\lisuppc#1#2#3#4{#1^{#2}#3\cdots#3#1^{#4}}%

\global\long\def\lisuppwout#1#2#3#4#5#6{#1^{#2}#3#4#3\wh{#1^{#6}}#3#4#3#1^{#5}}%

\global\long\def\lisubbwout#1#2#3#4#5#6{#1_{#2}#3#4#3\wh{#1}_{#6}#3#4#3#1_{#5}}%

\global\long\def\lisubwout#1#2#3#4{#1_{1}#2\dots#2\widehat{#1}_{#4}#2\dots#2#1_{#3}}%

\global\long\def\lisupwout#1#2#3#4{#1^{1}#2\dots#2\widehat{#1^{#4}}#2\dots#2#1^{#3}}%

\global\long\def\lisubwoutc#1#2#3#4{#1_{1}#2\cdots#2\widehat{#1}_{#4}#2\cdots#2#1_{#3}}%

\global\long\def\twp#1#2#3{\dee#1^{#2}\wedge\dee#1^{#3}}%

\global\long\def\thp#1#2#3#4{\dee#1^{#2}\wedge\dee#1^{#3}\wedge\dee#1^{#4}}%

\global\long\def\fop#1#2#3#4#5{\dee#1^{#2}\wedge\dee#1^{#3}\wedge\dee#1^{#4}\wedge\dee#1^{#5}}%

\global\long\def\idots#1{#1\dots#1}%
\global\long\def\icdots#1{#1\cdots#1}%

\global\long\def\norm#1{\|#1\|}%

\global\long\def\nonh{\heartsuit}%

\global\long\def\nhn#1{\norm{#1}^{\nonh}}%

\global\long\def\bigmid{\,\bigl|\,}%

\global\long\def\trps{^{{\scriptscriptstyle \textsf{T}}}}%

\global\long\def\testfuns{\mathcal{D}}%

\global\long\def\ntil#1{\tilde{#1}{}}%

\global\long\def\pis{y}%
\global\long\def\xo{\pis_{0}}%
\global\long\def\x{x}%

\global\long\def\pib{x}%
\global\long\def\bp{X}%
\global\long\def\ii{i}%
\global\long\def\ia{\alpha}%
\global\long\def\fp{y}%
\global\long\def\piv{v}%

\global\long\def\ib{i}%
\global\long\def\is{\alpha}%

\global\long\def\pbndo{\Gamma}%
\global\long\def\bndoo{\pbndo_{0}}%
 
\global\long\def\bndot{\pbndo_{t}}%
\global\long\def\intb{\inter{\body}}%
\global\long\def\bndb{\bdry\body}%

\global\long\def\cloo{\cl{\gO}}%

\global\long\def\nor{\mathbf{n}}%
\global\long\def\Nor{\mathbf{N}}%

\global\long\def\dA{\,\dee A}%

\global\long\def\dV{\,\dee V}%

\global\long\def\eps{\varepsilon}%

\global\long\def\tv{v}%
\global\long\def\av{u}%

\global\long\def\svs{\mathcal{W}}%
\global\long\def\vs{\mathbf{V}}%
\global\long\def\avs{\mathbf{U}}%
\global\long\def\affsp{\mathcal{A}}%
\global\long\def\man{\mathcal{M}}%
\global\long\def\odman{\mathcal{N}}%
\global\long\def\subman{\mathcal{V}}%
\global\long\def\pt{p}%

\global\long\def\vbase{e}%
\global\long\def\sbase{\mathbf{e}}%
\global\long\def\msbase{\mathfrak{e}}%
\global\long\def\vect{v}%
\global\long\def\dbase{\sbase}%

\global\long\def\chart{\varphi}%
\global\long\def\Chart{\Phi}%

\global\long\def\eucl{E}%

\global\long\def\mind{\alpha}%
\global\long\def\vb{W}%
\global\long\def\vbp{\pi}%

\global\long\def\vbt{\mathcal{E}}%
\global\long\def\fib{\vs}%
\global\long\def\vbts{W}%
\global\long\def\avb{U}%
\global\long\def\vbp{\xi}%

\global\long\def\chart{\vph}%
\global\long\def\vbchart{\Phi}%

\global\long\def\jetb#1{J^{#1}}%
\global\long\def\jet#1{j^{1}(#1)}%
\global\long\def\tjet{\tilde{\jmath}}%

\global\long\def\Jet#1{J^{1}(#1)}%

\global\long\def\jetm#1{j_{#1}}%

\global\long\def\coj{\mathfrak{d}}%

\global\long\def\alt{\mathfrak{A}}%

\global\long\def\pou{\eta}%

\global\long\def\ext{{\textstyle \bigwedge}}%
\global\long\def\forms{\Omega}%

\global\long\def\dotwedge{\dot{\mbox{\ensuremath{\wedge}}}}%

\global\long\def\vel{\theta}%

\global\long\def\Jac{\mathcal{J}}%

\global\long\def\contr{\mathbin{\raisebox{0.4pt}{\mbox{\ensuremath{\lrcorner}}}}}%
\global\long\def\fcor{\llcorner}%
\global\long\def\bcor{\lrcorner}%
\global\long\def\fcontr{\mathbin{\raisebox{0.4pt}{\mbox{\ensuremath{\llcorner}}}}}%

\global\long\def\lie{\mathcal{L}}%

\global\long\def\ssym#1#2{\ext^{#1}T^{*}#2}%

\global\long\def\sh{^{\sharp}}%

\global\long\def\nfo{\ext^{n}T^{*}\base}%
\global\long\def\dfs{\ext^{d}T^{*}\base}%
\global\long\def\dmfs{\ext^{d-1}T^{*}\base}%

\global\long\def\spc{\mathcal{S}}%
\global\long\def\sptm{\mathcal{E}}%
\global\long\def\evnt{e}%
\global\long\def\frame{\Phi}%

\global\long\def\timeman{\mathcal{T}}%
\global\long\def\zman{t}%
\global\long\def\dims{n}%
\global\long\def\m{\dims-1}%
\global\long\def\dimw{m}%

\global\long\def\wc{z}%

\global\long\def\fourv#1{\mbox{\ensuremath{\mathfrak{#1}}}}%

\global\long\def\pbform#1{\undertilde{#1}}%
\global\long\def\util#1{\raisebox{-5pt}{\ensuremath{{\scriptscriptstyle \sim}}}\!\!\!#1}%

\global\long\def\utilJ{\util J}%

\global\long\def\utilRho{\util{\rho}}%

\global\long\def\body{\mathcal{B}}%
\global\long\def\man{\mathcal{M}}%
\global\long\def\var{\mathcal{V}}%
\global\long\def\base{\mathcal{X}}%
\global\long\def\fb{\mathcal{Y}}%
\global\long\def\srfc{\mathcal{Z}}%
\global\long\def\dimb{n}%
\global\long\def\dimf{m}%
\global\long\def\afb{\mathcal{Z}}%

\global\long\def\bdry{\partial}%

\global\long\def\gO{\varOmega}%

\global\long\def\reg{\mathcal{R}}%
\global\long\def\bdrr{\bdry\reg}%

\global\long\def\bdom{\bdry\gO}%

\global\long\def\bndo{\partial\gO}%

\global\long\def\tpr{\vartheta}%

\global\long\def\mot{M}%
\global\long\def\vf{w}%
\global\long\def\const{h}%

\global\long\def\avf{u}%

\global\long\def\stn{\varepsilon}%
\global\long\def\djet{\chi}%

\global\long\def\jvf{\eps}%

\global\long\def\rig{r}%

\global\long\def\rigs{\mathcal{R}}%

\global\long\def\qrigs{\!/\!\rigs}%

\global\long\def\qd{\!/\,\!\kernel\diffop}%

\global\long\def\dis{\chi}%
\global\long\def\conf{\kappa}%
\global\long\def\invc{\hat{\conf}^{-1}}%
\global\long\def\dinvc{\hat{\conf}^{-1*}}%
\global\long\def\csp{\mathcal{Q}}%

\global\long\def\embds{\textrm{Emb}}%

\global\long\def\lc{A}%

\global\long\def\lv{\dot{A}}%
\global\long\def\alv{\dot{B}}%

\global\long\def\j{\mathop{\mathrm{j}}}%
\global\long\def\mapp{M}%
\global\long\def\J{J}%
\global\long\def\jex{\mathop{}\!\mathrm{j}}%

\global\long\def\fc{F}%
\global\long\def\load{f}%
\global\long\def\afc{g}%

\global\long\def\bfc{\mathbf{b}}%
\global\long\def\bfcc{b}%

\global\long\def\sfc{\mathbf{t}}%
\global\long\def\sfcc{t}%

\global\long\def\stm{\varsigma}%
\global\long\def\std{S}%
\global\long\def\tst{\sigma}%
\global\long\def\tstd{s}%
\global\long\def\st{\sigma}%
\global\long\def\vst{\varsigma}%
\global\long\def\vstd{S}%
\global\long\def\tstm{\sigma}%
\global\long\def\vstm{\varsigma}%

\global\long\def\stp{S_{P}}%
\global\long\def\slf{R}%

\global\long\def\crel{\Phi}%

\global\long\def\stmat{\tau}%

\global\long\def\gdiv{\bdry\textrm{iv\,}}%
\global\long\def\extjet{\mathfrak{d}}%

\global\long\def\smc#1{\mathfrak{#1}}%

\global\long\def\nhs{P}%
\global\long\def\nhsa{P}%
\global\long\def\nhsb{\underline{P}}%

\global\long\def\soc{Z}%

\global\long\def\sts{\varSigma}%
\global\long\def\spstd{\mathfrak{S}}%
\global\long\def\sptst{\mathfrak{T}}%
\global\long\def\spnhs{\mathcal{P}}%
\global\long\def\Ljj{\L{J^{1}(J^{k-1}\vb),\ext^{n}T^{*}\base}}%

\global\long\def\spsb{\text{\Large\ensuremath{\Delta}}}%

\global\long\def\ened{\mathfrak{w}}%
\global\long\def\energy{\mathfrak{W}}%

\global\long\def\ebdfc{T}%
\global\long\def\optimum{\st^{\textrm{opt}}}%
\global\long\def\scf{K}%

\global\long\def\grp{G}%
\global\long\def\gact{A}%
\global\long\def\gid{e}%
\global\long\def\gel{\ggm}%

\global\long\def\ael{\upsilon}%
\global\long\def\lal{\mathfrak{g}}%

\global\long\def\prop{P}%
\global\long\def\expr{P}%

\global\long\def\aprop{Q}%

\global\long\def\flux{\omega}%
\global\long\def\aflux{\psi}%

\global\long\def\fform{\tau}%

\global\long\def\dimn{n}%

\global\long\def\sdim{{\dimn-1}}%

\global\long\def\fdens{\phi}%

\global\long\def\pform{\varsigma}%
\global\long\def\vform{\beta}%
\global\long\def\sform{\tau}%
\global\long\def\flow{J}%
\global\long\def\n{\m}%
\global\long\def\cmap{\mathfrak{t}}%
\global\long\def\vcmap{\varSigma}%

\global\long\def\mvec{\mathfrak{v}}%
\global\long\def\mveco#1{\mathfrak{#1}}%
\global\long\def\mv#1{\mathfrak{#1}}%
\global\long\def\smbase{\mathfrak{e}}%
\global\long\def\spx{\simp}%
\global\long\def\il{l}%
\global\long\def\awe{\frown}%

\global\long\def\hp{H}%
\global\long\def\ohp{h}%

\global\long\def\hps{G_{\dims-1}(T\spc)}%
\global\long\def\ohps{G_{\dims-1}^{\perp}(T\spc)}%

\global\long\def\hyper{\mathcal{S}}%

\global\long\def\hpsx{G_{\dims-1}(\tspc)}%
\global\long\def\ohpsx{G_{\dims-1}^{\perp}(\tspc)}%

\global\long\def\fbun{F}%

\global\long\def\flowm{\Phi}%

\global\long\def\tgb{T\spc}%
\global\long\def\ctgb{T^{*}\spc}%
\global\long\def\tspc{T_{\pis}\spc}%
\global\long\def\dspc{T_{\pis}^{*}\spc}%

\global\long\def\fflow{\fourv J}%
\global\long\def\fvform{\mathfrak{b}}%
\global\long\def\fsform{\mathfrak{t}}%
\global\long\def\fpform{\mathfrak{s}}%
\global\long\def\lfc{\mathfrak{F}}%

\global\long\def\maxw{\mathfrak{g}}%
\global\long\def\frdy{\mathfrak{f}}%
\global\long\def\ptnl{\psi}%
\global\long\def\tptn{\Psi}%
\global\long\def\vptn{\mathfrak{a}}%
\global\long\def\mtst{\tstd_{M}}%
\global\long\def\mvst{\vstd_{M}}%

\global\long\def\sobp#1#2{W_{#2}^{#1}}%

\global\long\def\inner#1#2{\left\langle #1,#2\right\rangle }%

\global\long\def\fields{\sobp pk(\vb)}%

\global\long\def\bodyfields{\sobp p{k_{\partial}}(\vb)}%

\global\long\def\forces{\sobp pk(\vb)^{*}}%

\global\long\def\bfields{\sobp p{k_{\partial}}(\vb\resto{\bndo})}%

\global\long\def\loadp{(\sfc,\bfc)}%

\global\long\def\strains{\lp p(\jetb k(\vb))}%

\global\long\def\stresses{\lp{p'}(\jetb k(\vb)^{*})}%

\global\long\def\diffop{D}%

\global\long\def\strainm{E}%

\global\long\def\incomps{\vbts_{\yieldf}}%

\global\long\def\devs{L^{p'}(\eta_{1}^{*})}%

\global\long\def\incompsns{L^{p}(\eta_{1})}%

\global\long\def\testf{\mathcal{D}}%
\global\long\def\dists{\mathcal{D}'}%

\global\long\def\codiv{\boldsymbol{\partial}}%

\global\long\def\currof#1{\tilde{#1}}%

\global\long\def\chn{c}%
\global\long\def\chnsp{\mathbf{C}}%

\global\long\def\current{T}%
\global\long\def\curr{R}%

\global\long\def\curd{S}%
\global\long\def\curwd#1{\wh{#1}}%
\global\long\def\curnd#1{\wh{#1}}%

\global\long\def\contrf{{\scriptstyle \smallfrown}}%

\global\long\def\prodf{{\scriptstyle \smallsmile}}%

\global\long\def\form{\omega}%

\global\long\def\dens{\rho}%

\global\long\def\simp{s}%
\global\long\def\ssimp{\Delta}%
\global\long\def\cpx{K}%

\global\long\def\cell{C}%

\global\long\def\chain{B}%

\global\long\def\ach{A}%

\global\long\def\coch{X}%

\global\long\def\scale{s}%

\global\long\def\fnorm#1{\norm{#1}^{\flat}}%

\global\long\def\chains{\mathcal{A}}%

\global\long\def\ivs{\boldsymbol{U}}%

\global\long\def\mvs{\boldsymbol{V}}%

\global\long\def\cvs{\boldsymbol{W}}%

\global\long\def\ndual#1{#1'}%

\global\long\def\nd{'}%

\global\long\def\cee#1{C^{#1}}%

\global\long\def\lone{\{L^{1}\}}%

\global\long\def\linf{L^{\infty}}%

\global\long\def\lp#1{L^{#1}}%

\global\long\def\ofbdo{(\bndo)}%

\global\long\def\ofclo{(\cloo)}%

\global\long\def\vono{(\gO,\rthree)}%

\global\long\def\lomu{\{L^{1,\mu}\}}%
\global\long\def\limu{L^{\infty,\mu}}%
\global\long\def\limub{\limu(\body,\rthree)}%
\global\long\def\lomub{\lomu(\body,\rthree)}%

\global\long\def\vonbdo{(\bndo,\rthree)}%
\global\long\def\vonbdoo{(\bndoo,\rthree)}%
\global\long\def\vonbdot{(\bndot,\rthree)}%

\global\long\def\vonclo{(\cl{\gO},\rthree)}%

\global\long\def\strono{(\gO,\reals^{6})}%

\global\long\def\sob{\{W_{1}^{1}\}}%

\global\long\def\sobb{\sob(\gO,\rthree)}%

\global\long\def\lob{\lone(\gO,\rthree)}%

\global\long\def\lib{\linf(\gO,\reals^{12})}%

\global\long\def\ofO{(\gO)}%

\global\long\def\oneo{{1,\gO}}%
\global\long\def\onebdo{{1,\bndo}}%
\global\long\def\info{{\infty,\gO}}%

\global\long\def\infclo{{\infty,\cloo}}%

\global\long\def\infbdo{{\infty,\bndo}}%
\global\long\def\lobdry{\lone(\bdry\gO,\rthree)}%

\global\long\def\ld{LD}%

\global\long\def\ldo{\ld\ofO}%
\global\long\def\ldoo{\ldo_{0}}%

\global\long\def\trace{\gamma}%
\global\long\def\dtrace{\delta}%
\global\long\def\gtrace{\beta}%

\global\long\def\pr{\proj_{\rigs}}%

\global\long\def\pq{\proj}%

\global\long\def\qr{\,/\,\reals}%

\global\long\def\aro{S_{1}}%
\global\long\def\art{S_{2}}%

\global\long\def\mo{m_{1}}%
\global\long\def\mt{m_{2}}%

\global\long\def\ebdfc{T}%

\global\long\def\mini{\Omega}%
\global\long\def\optimum{s^{\mathrm{opt}}}%
\global\long\def\scf{K}%
\global\long\def\opsf{\st^{\mathrm{opt}}}%
\global\long\def\doptimum{s^{\opt,{\scriptscriptstyle D}}}%
\global\long\def\loptimum{s^{\opt,{\scriptscriptstyle \mathcal{M}}}}%

\global\long\def\fsubs{M}%

\global\long\def\yieldc{B}%

\global\long\def\yieldf{Y}%

\global\long\def\trpr{\pi_{P}}%

\global\long\def\devpr{\pi_{\devsp}}%

\global\long\def\prsp{P}%

\global\long\def\devsp{D}%

\global\long\def\ynorm#1{\|#1\|_{\yieldf}}%

\global\long\def\colls{\Psi}%

\global\long\def\aro{S_{1}}%
\global\long\def\art{S_{2}}%

\global\long\def\mo{m_{1}}%
\global\long\def\mt{m_{2}}%

\global\long\def\trps{^{\mathsf{T}}}%

\global\long\def\hb{^{\mathrm{hb}}}%

\global\long\def\yieldst{s_{Y}}%

\global\long\def\yieldc{B}%

\global\long\def\lcap{C}%

\global\long\def\yieldf{Y}%

\global\long\def\sphpr{\pi_{P}}%

\global\long\def\devpr{\pi_{\devsp}}%

\global\long\def\prsp{P}%

\global\long\def\devsp{D}%

\global\long\def\ynorm#1{\|#1\|_{\yieldf}}%

\global\long\def\colls{\Psi}%

\global\long\def\cone{Q}%
\global\long\def\fpr{\Pi}%
\global\long\def\fprd{\fpr_{\devsp}}%
\global\long\def\fprp{\fpr_{\prsp}}%
\global\long\def\find{I_{\devsp}}%
\global\long\def\finp{I_{\prsp}}%
\global\long\def\fnorm#1{\norm{#1}_{\devsp}}%

\global\long\def\rig{r}%
\global\long\def\rigs{\mathcal{R}}%
\global\long\def\qrigs{\!/\!\rigs}%
\global\long\def\anv{\omega}%
\global\long\def\I{I}%
\global\long\def\mone{M_{1}}%

\global\long\def\bd{BD}%

\global\long\def\po{\proj_{0}}%
\global\long\def\normp#1{\norm{#1}'_{\ld}}%

\global\long\def\ssx{S}%

\global\long\def\smap{s}%

\global\long\def\smat{\chi}%

\global\long\def\sx{e}%

\global\long\def\snode{P}%
\global\long\def\newmacroname{}%

\global\long\def\elem{e}%

\global\long\def\nel{L}%

\global\long\def\el{l}%

\global\long\def\gr{g}%
\global\long\def\ngr{G}%

\global\long\def\eldof{\alpha}%

\global\long\def\glbs{\psi}%

\global\long\def\ipln{\phi}%

\global\long\def\ndof{D}%

\global\long\def\dof{d}%

\global\long\def\nldof{N}%

\global\long\def\ldof{n}%

\global\long\def\lvf{\chi}%

\global\long\def\amat{A}%
\global\long\def\bmat{B}%

\global\long\def\subsp{\mathcal{M}}%
\global\long\def\zerofn{Z}%

\global\long\def\snomat{E}%

\global\long\def\femat{E}%

\global\long\def\tmat{T}%

\global\long\def\fvec{f}%

\global\long\def\snsp{\mathcal{S}}%

\global\long\def\slnsp{\Phi}%
\global\long\def\dslnsp{\Phi^{{\scriptscriptstyle D}}}%

\global\long\def\ro{r_{1}}%

\global\long\def\rtwo{r_{2}}%

\global\long\def\rth{r_{3}}%

\global\long\def\fmax{M}%

\global\long\def\dform{\psi}%

\global\long\def\srfc{\mathcal{S}}%

\global\long\def\semib{\mathrm{SB}}%

\global\long\def\tm#1{\overrightarrow{#1}}%
\global\long\def\tmm#1{\underrightarrow{\overrightarrow{#1}}}%

\global\long\def\itm#1{\overleftarrow{#1}}%
\global\long\def\itmm#1{\underleftarrow{\overleftarrow{#1}}}%

\global\long\def\ptrac{\mathcal{P}}%

\global\long\def\nh#1{\hat{#1}}%
\global\long\def\nj{\hat{\jmath}}%
\global\long\def\nJ{\hat{J}}%
\global\long\def\rin#1{\mathfrak{#1}}%
\global\long\def\npi{\hat{\pi}}%
\global\long\def\rp{\rin p}%
\global\long\def\rq{\rin q}%
\global\long\def\rr{\rin r}%

\global\long\def\xty{(\base,\fb)}%
\global\long\def\xts{(\base,\spc)}%
\global\long\def\r{r}%
\global\long\def\ntm{(\reals^{n},\reals^{m})}%

\global\long\def\mtn{e}%
\global\long\def\sppp{\lambda}%

\global\long\def\mtsp{\mathscr{E}}%

\global\long\def\disp{g}%
\global\long\def\diffs{G}%

\global\long\def\bv{BV}%

\title[Incompatible-Compatible Decomposition]{Continuum Kinematics with Incompatible-Compatible Decomposition}
\author{Vladimir Goldshtein$\vphantom{N^{2}}^{1}$, Paolo Maria Mariano$\vphantom{N^{2}}^{2}$,
Domenico Mucci$\vphantom{N^{2}}^{3}$,\\
 and Reuven Segev$\vphantom{N^{2}}^{4}$}
\address{}
\begin{abstract}
We present a framework for the kinematics of a material body undergoing
anelastic deformation. For such processes, the material structure
of the body, as reflected by the geometric structure given to the
set of body points, changes. The setting we propose may be relevant
to phenomena such as plasticity, fracture, discontinuities and non-injectivity
of the deformations. In this framework, we construct an unambiguous
decomposition into incompatible and compatible factors which includes
the standard elastic-plastic decomposition in plasticity.
\end{abstract}

\date{\today\\[2mm]
$^1$ Department of Mathematics, Ben-Gurion University of the Negev, Israel. Email: vladimir@bgu.ac.il\\
$^2$ DICEA - University of Florence via Santa Marta 3, I-50139 Firenze, Italy. Email: paolomaria.mariano@unifi.it\\
$^3$ DSMFI - UNiversità di Parma Parco Area delle Scienze 53/A, I-43134 Parma, Italy. Email: domenico.mucci@unipr.it\\
$^4$ Department of Mechanical Engineering, Ben-Gurion University of the Negev, Israel. Email: rsegev@post.bgu.ac.il}

\maketitle

\section{Introduction}

The elastic-plastic decomposition of the deformation gradient, $F$,
into an ``elastic'' factor, $F^{e}$, and a ``plastic'' factor, $F^{p}$,
as $F=F^{e}F^{p}$, was introduced in 1960 by Kröner \cite{Kr60}
and in 1967 by Lee \& Liu \cite{Lee1967,L69}, and has been used and
studied extensively since then. For a comprehensive review of the
subsequent work see \cite{SH98}, and for more recent work see for
example \cite{AA20}, \cite{Cas17}, \cite{CD20}, \cite{GGY21},
\cite{M13}, \cite{Miehe14}, \cite{Mielke03}, \cite{YS20}, \cite{YS23}.
The plastic factor is viewed as the tensor field needed in order to
release the residual stresses in the reference unloaded configuration
of the body. The incompatibility reflects the macroscopic description
of the existence of defects in the material. From another point of
view, \eg,  \cite{Miehe1998}, \cite{YS20}, \cite{YS23}, it is
impossible to embed isometrically the body with the stress-free metric
tensor in a $3$-dimensional Euclidean space. Under such interpretations,
the elastic factor, $F^{e}$, describes the incompatible packing of
the stress free body elements to restore compatibility to $F=F^{e}F^{p}$.

Another view on the elastic-plastic decomposition is proposed by C.
Reina and S. Conti \cite{RC14,RSC16,R_et_al_18} where starting from
a perfect lattice, $F^{p}$ corresponds to a change of material structure
of the lattice\textemdash a change in the topology\textemdash while
$F^{e}$ corresponds to the placement of the defected structure in
space.

The following exemplifies some approaches to motivate the elastic-plastic
decomposition.
\begin{itemize}
\item By looking at lattices and considering a notion of defectiveness defined
referring to invariant peculiar features with respect to the action
of diffeomorphisms, G. Parry arrived at a multiplicative decomposition
that involves two factors of the type $F^{p}$, one preceding $F^{e}$,
the other following it \cite{Pa04}. \smallskip{}
\item In crystals, slips may occur along special planes and are a source
of unrecoverable strain determined by the slip of dislocations. Across
such planes deformations suffer jumps of finite amplitude. A way to
model the circumstance is selecting deformations to be special maps
of bounded variations ($SBV$-maps). Such maps admit a distributional
derivative that is a measure with an additive decomposition into a
bulk part, which is absolutely continuous with respect to the Lebesgue
volume measure, and a singular component concentrated over a rectifiable
set with $m-1$ Hausdorff's measure, where $m$ is the dimension of
the domain. The multiplicative decomposition $F=F^{e}F^{p}$ emerges
naturally, as shown by C. Reina and S. Conti \cite{RC14} (see also
\cite{RSC16,R_et_al_18} and \cite{MM22}; the latter reference accounts
for possible volumetric plastic changes in the $SBV$ setting). In
this view, $F^{p}$ is a measure, while $F^{e}$ a gradient, taken
with positive determinant. In \cite{R_et_al_18} the plastic deformation
is shown to follow from a coarse-graining procedure from the lattice
mesoscopic description.
\end{itemize}
Here, we propose a framework which shares similarities with these
last two approaches. Like \cite{R_et_al_18}, we view the plastic
factor as assignment of topological structure to the body. Similarly
to \cite{Pa04}, we take material structure to be invariant under
a subgroup of the group of diffeomorphisms.

\medskip{}

When we refer to continuum mechanics, we commonly say that it is the
qualitative and quantitative description of the way tangible bodies
react under external actions. The definition requires clarification
of the essential nature of what we call a \emph{body}: this is a conceptual
choice, we need to make\textemdash and we do this even unconsciously\textemdash in
building up mathematical models of natural phenomena.

In basic treatises on continuum mechanics\textemdash mainly those
emerging from the work of C.~A.~Truesdell's school\textemdash a
body is taken to be a set of not otherwise specified \emph{material
elements}, presumed to be endowed with the structure of a finite-dimensional
manifold \cite{Noll1959} \cite{T77}, \cite{TT60}, \cite{TN65},
\cite{N73}, \cite{Sil}. In particular, in \cite{Noll1959}, the
manifold structure of body is manifested by its configurations in
the $3$-dimensional Euclidean space.

This setting may be extended to the situation where the physical space
is modeled as a general $n$-dimensional manifold, $\spc$. Such a
generalization may be motivated, for example, by considering small
scale interactions, or microstructure. In this case, configurations
will be valued in a fiber bundle over a Euclidean space \cite{C89},
\cite{Seg94}, \cite{M02}, \cite{M16}.

Anelasticity is associated with changes of material structure\textemdash the
topological or geometric structure of the material body. Thus, one
has to make a clear distinction between a body, which has a certain
manifold structure, and the collection of points that the body comprises.
To identify the object, the material structure of which may change
in an anelastic process, we use the term \emph{protobody}. The various
material structures that a protobody may attain in anelastic processes
are referred to as \emph{embodiments.} Each embodiment of a protobody
should be a body of continuum mechanics. (See \cite{Seg96} and \cite{SegEp96},
where analogous notions are presented for theories of growing bodies.)

The configuration space, $\csp$, of a protobody in space, should
contain all the configurations in space at all possible embodiments
of the body. To each configuration $\conf$ of the protobody in space,
there corresponds an embodiment $\mtn$ of the protobody. However,
it is expected that for each embodiment there will be a subset of
configurations of the protobody.

Thus, we say that two configurations, $\conf_{1}$ and $\conf_{2}$,
of the protobody correspond to the same embodiment if there is a diffeomorphism,
$\disp_{21}$, of the space manifold such that $\conf_{2}=\disp_{21}\comp\conf_{1}$.
This induces an equivalence relation on $\csp$, for which an embodiment
is an equivalence class, and the embodiment space, $\mtsp$, is the
quotient set.

Next, we show that an embodiment, $\mtn$, may be represented as a
topological space, $\body_{\mtn}$, the elements of which are the
body points associated with that embodiment. It is noted that we do
not restrict configurations of the protobody to be injective. As a
result, the topological spaces $\body_{\mtn_{1}}$ and $\body_{\mtn_{2}}$,
for two distinct embeddings need not comprise the same material points.
Finally, each configuration $\conf:\body\to\spc$ is factored in the
form $\conf=\conf_{e}\comp\conf_{a\mtn}$, where $\conf_{a\mtn}:\body\to\body_{\mtn}$
and $\conf_{\mtn}:\body_{\mtn}\to\spc$, the analog of the elastic-plastic
decomposition. There is no ambiguity in the decomposition.

This general framework makes it possible to represent discontinuous
and non-injective configurations of the protobody in space, modeling
phenomena such as fracture and destruction of material points.

To consider phenomena such as plasticity for which the the elastic-plastic
decomposition applies to the deformation gradient, we have to be more
specific. Thus, we substitute for the protobody the tangent bundle
$T\body$ of a manifold $\body$, representing the perfect crystallographic
structure of the body. A configuration is represented by a vector
bundle morphism $\conf:T\body\to T\spc$. Incompatibility occurs when
$\conf$ is not the tangent mapping of the base map $\ul{\conf}:\body\to\spc$.
We say that two configurations, $\conf_{1}$ and $\conf_{2}$, correspond
to the same embodiment when there is a diffeomorphism, $\disp_{21}$,
of space such that $\conf_{2}=T\disp_{21}\comp\conf_{1}$. We show
that an embodiment is represented by a vector bundle, representing
the ``dislocated'' material structure, and the elastic-plastic decomposition
of vector bundle configurations of a protobody follows.

Section \ref{sec:The-Basic-Framework} below outlines the general
framework we propose for the kinematics of elastic-anelastic processes.
Section \ref{sec:Groupoids} describes some of the notions of the
general framework in terms of groupoids. This section may be skipped
without interrupting the rest of the text. Section \ref{sec:Infinitesimal-Incompatibility}
considers the case where configurations are vector bundle morphisms
defined on the tangent bundle of a manifold. As mentioned above, the
tangent bundle represents a solid body together with its microstructure.
Section \ref{sec:Dislocated} specializes the foregoing one to the
case where the base mapping of the vector bundle morphisms representing
the configurations, are embeddings. This situation is analogous, in
the geometry of differentiable manifolds, to the classical elastic-plastic
decomposition described above. Finally, in Section \ref{sec:Dislocated-Quasicrystals},
we make some comments as to the relevance of the proposed framework
to quasicrystals.

\section{\label{sec:The-Basic-Framework}The Basic Framework}

\subsection{Basic definitions}

Let $\body$ be a set, which we view as a collection of material points,
and refer to it as a \emph{protobody}. We do not assume at this stage
that $\body$ has any particular structure. As a standard example,
the protobody may be represented by a bounded open subset of $\rthree$.

The \emph{physical space} is modeled by an $n$-dimensional oriented
differentiable manifold $\spc$. In traditional formulation of continuum
mechanics $\spc$ is modeled as a $3$-dimensional Euclidean space.

The \emph{configuration space}, $\csp$, of the protobody is assumed
to be a given class of mappings of the protobody into the space manifold.
A generic element of $\csp$ is denoted as $\conf:\body\to\spc$.
For example, if $\body$ is a bounded and connected open subset of
$\rthree$ and $\spc=\rthree$, one may consider the case where $\csp=\bv(\body,\rthree)$,
or $\csp=SBV^{p}(\body,\rthree)$, with appropriate $p$, when discontinuities
of the deformation distributional derivative (a measure, indeed) do
not include a Cantor set and the absolutely continuous part with respect
to the Lebesgue measure is endowed with $L^{p}$ density.

In the rest of the text, we refer to \emph{bi-Lipschitz, oriented
diffeomorphisms} simply as \emph{diffeomorphisms}. On a differentiable
manifold, bi-Lipschitz mappings may be defined using a Riemannian
metric. The class of bi-Lipschitz mappings is invariant under the
particular choice of a Riemannian metric.
\begin{defn}
Let $\diffs$ be a subgroup of the group of diffeomorphisms of $\spc$.
We say that $\conf_{1},\conf_{2}\in\csp$ are \emph{compatible} if
there is a diffeomorphism $\disp_{21}\in\diffs$ such that 
\begin{equation}
\conf_{2}=\disp_{21}\comp\conf_{1}.
\end{equation}
In such a case, we refer to $\disp_{21}$ as a (compatible) \emph{displacement
}and we write $\conf_{1}\sim\conf_{2}$.
\end{defn}

\begin{rem}
As a possible generalization of this definition one may consider a
group of bijective mappings $\spc\to\spc$ that are not necessarily
smooth. This may lead to a relaxed definition of the compatible (elastic)
factor of the decomposition. For example, one may consider a subgroup
of the group of bi-Lipschitz mappings on $\spc$.
\end{rem}

Evidently, compatibility is an equivalence relation, which justifies
the notation we adopt.
\begin{defn}
The quotient space, 
\begin{equation}
\mtsp:=\csp/\sim
\end{equation}
will be referred to as the \emph{space of material structures} or
the \emph{embodiment space}. An element $\mtn\in\mtsp$ represents
a \emph{material structure }or an \emph{embodiment.}
\end{defn}

Thus, we have a natural projection
\begin{equation}
\pi_{\mtsp}:\csp\tto\mtsp,\qquad\conf\lmt[\conf],
\end{equation}
where $[\conf]$ denotes the equivalence class of $\conf$.

\subsection{The structure induced by an embodiment}

Any embodiment induces a topological space. In fact, let $\mtn\in\mtsp$
be an embodiment, and define
\begin{equation}
A_{\mtn}:=\coprod_{\conf\in\mtn}\image\conf.\label{eq:disjoint_union}
\end{equation}
An element $a\in A_{\mtn}$ is represented by $(y,\conf)$ where $y\in\image\conf\subset\spc$,
$\conf\in\mtn$.

Consider the following relation on $A_{\mtn}$. We say that 
\begin{equation}
a_{1}=(y_{1},\conf_{1})\sim_{\mtn}a_{2}=(y_{2},\conf_{2})\qquad\text{if}\qquad y_{2}=\disp_{21}(y_{1})\label{eq:Eq-Rel-e}
\end{equation}
for $\disp_{21}\in\diffs$ satisfying $\conf_{2}=\disp_{21}\comp\conf_{1}$.
By the definition of $\mtsp$, such a diffeomorphism exists. Evidently,
$\sim_{\mtn}$ is an equivalence relation. The equivalence class of
$a\in A_{\mtn}$ will be denoted as $[a]_{\mtn}$. The quotient space
$A_{\mtn}/\sim_{\mtn}$ will be denoted by $\body_{\mtn}$, so that
we have a natural projection
\begin{equation}
\pi_{\mtn}:A_{\mtn}\tto\body_{\mtn}=A_{\mtn}/\sim_{\mtn}.\label{eq:B_e}
\end{equation}
An element $x\in\body_{\mtn}$ is interpreted as a \emph{body point}
contained in the embodiment $\mtn$ of the protobody. The set $\body_{\mtn}$
is interpreted as the set of body points contained in the embodiment
$\mtn$. We may refer to $\body_{\mtn}$ as the \emph{body structure}
induced by the embodiment $\mtn$.

Let $\mtn\in\mtsp$ be an embodiment, and let $\conf\in\mtn$. We
have a natural mapping 
\begin{equation}
\pi_{\mtn\conf}:\image\conf\tto\body_{\mtn},\qquad y\lmt[(y,\conf)]_{\mtn}.\label{eq:pi_ek}
\end{equation}
The mapping $\pi_{\mtn\conf}$ is clearly a bijection. The body point
$x=\pi_{\mtn\conf}(y)$ occupies the location $y\in\spc$ at the configuration
$\conf$.

Let $\conf\in\mtn$ be a configuration. Then, $\image\conf$ has the
subspace topology it inherits from the manifold $\spc$. If $\conf_{1},\conf_{2}\in\mtn$
so that $\conf_{2}=\disp_{21}\comp\conf_{1}$, then, $\disp_{21}\resto{\image\conf_{1}}:\image\conf_{1}\to\image\conf_{2}$
is a homeomorphism. This induces a topology on $\body_{\mtn}$ by
defining a subset $U\subset\body_{\mtn}$ to be open if for some $\conf\in\mtn$,
and an open subset $U_{\conf}\subset\image\conf$, 
\begin{equation}
U=\pi_{\mtn\conf}(U_{\conf}).
\end{equation}
The topology is well defined, and is independent of the choice of
$\conf\in\mtn$. Moreover, with this topology, $\pi_{\mtn\conf}:\image\conf\to\body_{\mtn}$
is a homeomorphism for each $\conf\in e$,

In all practical cases, $\image\conf$ will be a topological submanifold
of $\spc$. If $\image\conf$ is an oriented differentiable submanifold
of $\spc$ for some $\conf\in\mtn$, this applies to all other $\conf'\in\mtn$.
In this case, for $\conf_{1},\conf_{2}\in\mtn,$ $\disp_{21}\resto{\image\conf_{1}}:\image\conf_{1}\to\image\conf_{2}$
is a diffeomorphism. A procedure analogous to the one above induces
an oriented manifold structure on $\body_{\mtn}$ for which $\pi_{\mtn\conf}$
is a diffeomorphism.

\subsection{The incompatible-compatible decomposition}

Let $\conf:\body\to\spc$ be a configuration, $\incl_{\conf}:\image\conf\to\spc$
the natural inclusion, and $\mtn=\pi(\conf)\in\mtsp$ the induced
embodiment. Since $\pi_{\mtn\conf}:\image\conf\to\body_{\mtn}$ is
a homeomorphism, the same applies to $\pi_{\mtn\conf}^{-1}$ and we
can define
\begin{equation}
\conf_{\mtn}:=\incl_{\conf}\comp\pi_{\mtn\conf}^{-1}:\body_{\mtn}\tto\spc.
\end{equation}
The mapping $\conf_{\mtn}$ is interpreted as the \emph{compatible
factor of the configuration} $\conf$. It is the analog of the ``elastic''
factor of the ``plastic-elastic'' decomposition. Clearly, the compatible
factor of the configuration is a continuous injection into $\spc$.
In case $\image\conf$ is an oriented submanifold of $\spc$, and
we use the induced differentiable structure on $\body_{e}$, $\conf_{\mtn}$
is an embedding.

For the same variables as above, consider the mapping 
\begin{equation}
\conf_{a\mtn}:=\pi_{\mtn\conf}\comp\conf:\body\tto\body_{\mtn}.
\end{equation}
Then,
\begin{equation}
\conf=\conf_{\mtn}\comp\conf_{a\mtn},
\end{equation}
which is the incompatible-compatible decomposition (see the diagram
below).
\begin{equation}
\xymatrix{ &  & \image\conf\subset\spc\ar[d]_{\pi_{\mtn\conf}}\\
\body\ar[rr]_{\conf_{a\mtn}}\ar[rru]^{\conf} &  & \body_{\mtn}\ar@/_{1pc}/[u]_{\conf_{\mtn}}.
}
\end{equation}

The relevance of the decomposition follows from the following property.
\begin{lem}
Let $\conf\in Q$ be a configuration. Then, $\conf_{a\mtn}$ depends
only on $\mtn=\pi_{\mtsp}(\conf)\in\mtsp$.
\end{lem}

\begin{proof}
Let $X\in\body$ and $\conf'\sim\conf$. We have to show that $\conf'_{a\mtn}(X)=\conf{}_{a\mtn}(X)$.
Since $\conf'\sim\conf$, there is a diffeomorphism, $\disp\in\diffs$,
such that $\conf'=\disp\comp\conf$. Hence, $\conf'(X)=\disp(\conf(X))$.
By the definition (\ref{eq:Eq-Rel-e}), 
\begin{equation}
(\conf'(X),\conf')\sim_{\mtn}(\conf(X),\conf).
\end{equation}
The definition of $\pi_{\mtn\conf}$ in (\ref{eq:pi_ek}), implies
now that 
\begin{equation}
\pi_{\mtn\conf'}(\conf'(X))=\pi_{\mtn\conf}(\conf(X)).
\end{equation}
\end{proof}
When $\body$ has an oriented manifold structure, and $\conf$ is
an oriented embedding, $\conf_{a\mtn}$ is a diffeomorphism. Thus,
in such a case, one can identify $\body$ with $\body_{\mtn}$, \ie, 
$\conf_{a\mtn}$ reduces to an identity. In general, $\body$ has
no structure, and compatibility of $\conf_{a\mtn}:\body\to\body_{\mtn}$
in the standard sense of continuum mechanics cannot even be defined.
Intermediate situations, where the incompatibility of $\conf_{a\mtn}$
is significant and well defined, are considered below.

\subsection{\label{subsec:Reference-configurations}Reference configurations}

A right inverse of $\pi_{\mtsp}$,
\begin{equation}
r:\mtsp\to\csp,
\end{equation}
may be interpreted as a system of reference configurations for the
embodiments of $\body$. In other words, $r(\mtn)$ is the reference
configuration for the material structure $\mtn$ in space.

When a system of reference configurations is given, one may accept
the identification $\body_{\mtn}:=\image r(\mtn)$. In such a situation,
for some $\conf\in\csp$, the deformation $\conf_{\mtn}$ may be identified
with the restriction, 
\begin{equation}
\disp_{\conf,r(\mtn)}\resto{\image r(\mtn)}:\image r(\mtn)\tto\spc,
\end{equation}
 of the diffeomorphism $\disp_{\conf,r(\mtn)}$ of $\spc$ to $\image r(\mtn)$.
The mapping $\conf_{a\mtn}$ is identified in this case with $r(\mtn)$.
 In case a system of reference configurations is not given, the term
``intermediate'' \emph{configuration} does not describe the situation
appropriately because $\conf_{a\mtn}$ is valued in the abstract manifold
$\body_{\mtn}$ and not in space.

Also, although commonly used as a terminology in modeling plasticity,
an \textquotedblleft intermediate\textquotedblright{} configuration
intended as a \emph{global} configuration obtained by rearranging
in incompatible way the material texture, is in general not available.
In analyzing strain, we essentially have a local description of the
incompatibility due to the rearrangement of the material structure.
We locally map the tangent space at a point in some configuration
into an \textquotedblleft intermediate\textquotedblright{} space,
and the mapping is incompatible in the sense of being not congruent.
This circumstance leads us to model incompatibility of tangent-plane-neighborhoods,
as we do in the Section \ref{sec:Infinitesimal-Incompatibility}.

\section{\label{sec:Groupoids}The Groupoid Point of View}

This section describes how some of the foregoing structure can be
described and generalized using the language of groupoid theory. It
is of formal nature and may be skipped without interrupting the reading
of the following sections. Roughly speaking, a groupoid consists of
a collection of elements and a collection of arrows between pairs
of elements. In particular, not all pairs of elements are connected
by an arrow. Arrows can be composed and inverted in a consistent way.

In our situation, we have a set $\csp$, and a set $\gG$ containing
mappings. The configurations in $\csp$ are referred to as \emph{objects}
in the language of groupoid theory and the elements of $\gG$ are
referred to as \emph{morphisms}. Each morphism, $\ggm$, is associated
with a configuration $\conf_{1}=\ga(\ggm)$ and a configuration $\conf_{2}=\gb(\ggm)$,
and $\ggm$ represents a mapping (a restriction of a diffeomorphism
of $\spc$ in the case considered above) $\image\conf_{1}\to\image\conf_{2}$.
In such a case, we write $\conf_{2}=\ggm\conf_{1}$.

The mappings $\ga:\gG\to\csp$ and $\gb:\gG\to\csp$ are referred
to as the \emph{source map} and \emph{target map}, respectively. Note
that here, we do not require that the diffeomorphism be extended to
a diffeomorphism of $\spc$.

It is emphasized that not any pair of elements of $\csp$ are the
source and target of some morphism. In general, there are pairs $\conf_{1},\conf_{2}\in\csp$,
representing incompatible configurations of the body, for which there
is no connecting morphism. For the case where there is a morphism
$\ggm$ such that $\conf_{2}=\ggm\conf_{1}$, we have written $\conf_{2}\sim\conf_{1}$.

The morphisms satisfy the following properties.
\begin{enumerate}
\item $\ga$ and $\gb$ are surjective.
\item For \emph{composable} morphisms $\ggm_{1},\ggm_{2}\in\gG$, that is
$\gb(\ggm_{1})=\ga(\ggm_{2})$, there is a composition $\ggm_{2}\cdot\ggm_{1}\in\gG$
such that 
\begin{equation}
\ga(\ggm_{2}\cdot\ggm_{1})=\ga(\ggm_{1}),\qquad\gb(\ggm_{2}\cdot\ggm_{1})=\gb(\ggm_{2}).
\end{equation}
\item The composition is associative, so for three composable morphisms
\begin{equation}
\ggm_{3}\cdot(\ggm_{2}\cdot\ggm_{1})=(\ggm_{3}\cdot\ggm_{2})\cdot\ggm_{1}.
\end{equation}
\item For each $\conf\in\csp$ there is a morphism $\eps_{\conf}$\textemdash corresponding
to the identity mapping $\image\conf\to\image\conf$\textemdash such
that $\ga(\eps_{\conf})=\gb(\eps_{\conf})=\conf$, and 
\begin{equation}
\ggm\cdot\eps_{\ga(\ggm)}=\ggm=\eps_{\gb(\ggm)}\cdot\ggm\fall\ggm\in\gG.
\end{equation}
\item For each $\ggm\in\gG$, there is $\ggm^{-1}\in\gG$, in our case the
inverse mapping, such that 
\begin{equation}
\ggm^{-1}\cdot\ggm=\eps_{\ga(\ggm)}.\qquad\ggm\cdot\ggm^{-1}=\eps_{\gb(\ggm)}.
\end{equation}
\end{enumerate}
The restrictions of diffeomorphisms of $\spc$, of the type $\image\conf_{1}\to\image\conf_{2}$,
evidently satisfy these conditions. This implies that $\gG$ is a
\emph{groupoid} over $\csp$ and we express this as $\gG\rightrightarrows\csp$.
(See \cite{DLE21} for the theory of groupoids and some of its applications
to continuum mechanics.)

In the language of groupoid theory, the set
\begin{equation}
\mathcal{O}(\conf):=\gb(\ga^{-1}\{\conf\})=\ga(\gb^{-1}\{\conf\})\subset\csp
\end{equation}
is referred to as the \emph{orbit} of $\conf$. However, in our notation,
the orbit is simply the equivalence class of $\conf$\textemdash an
embodiment of the body. The quotient space\textemdash the embodiment
space in our application\textemdash is referred to as the \emph{orbit
space}.

\medskip{}

Another groupoid structure corresponds to the construction of the
set $\body_{\mtn}$ for some given embodiment $\mtn\in\mtsp$. The
space of objects in this case is $A_{\mtn}$ defined above, so that
an object is represented by $(y,\conf)$, $y\in\image\conf$. A morphism
$\gd$ sends $(y_{1},\conf_{1})$ to $(y_{2},\conf_{2})$, where $\conf_{2}=\ggm\conf_{1}$,
$\ggm\in\gG$. Evidently, given $\conf_{1}$ and $\conf_{2}$, with
$\conf_{2}=\ggm\conf_{1}$, there is only a single $y_{2}\in\image\conf_{2}$
such that $(y_{2},\conf_{2})=\gd(y_{1},\conf_{1})$. The resulting
groupoid will be denoted as $\gG_{\mtn}\rightrightarrows A_{\mtn}$.
Thus, for the case described, 
\begin{equation}
\ga(\gd)=(y_{1},\conf_{1}),\qquad\gb(\gd)=(y_{2},\conf_{2}).
\end{equation}
Using the language of groupoids, a point $x\in\body_{\mtn}$ is an
orbit in $\gG_{\mtn}$ and $\body_{\mtn}$ is the orbit space of $\gG_{\mtn}$.

\section{\label{sec:Infinitesimal-Incompatibility}Infinitesimal Incompatibility}

For the case where $\body$ has an oriented manifold structure, a
natural bundle morphism is associated with $\conf:\body\to\spc$;
it is the tangent map, $T\conf$, from $T\body$ to $T\spc$. When
we aim at describing elastic-plastic phenomena, we need to model incompatibility
of tangent planes\textemdash ``infinitesimal neighborhoods'' of material
points\textemdash that may occur even in the case of smooth placements
of the material points in space. To account for such incompatibility,
we need to extend the view described so far. Specifically, we will
no longer consider $\conf$ as a map from $\body$ to $\spc$, rather
we take $\conf$ itself as a vector bundle morphism from $T\body$
to $T\spc$.

\subsection{Infinitesimal configurations and embodiments}

We specialize the setting of Section \ref{sec:The-Basic-Framework}
by replacing first the protobody general set $\body$ by the tangent
bundle, $T\body$ of an oriented manifold $\body$, where we have
the projection
\begin{equation}
\tau_{\body}:T\body\tto\body.
\end{equation}
The tangent space $T_{X}\body$ at $X\in\body$ represents the ``infinitesimal
neighborhood'' of $X$.

The configuration space $\csp$ is a family of vector bundle morphisms
\begin{equation}
\conf:T\body\tto T\spc.
\end{equation}
For $\conf\in\csp$, 
\begin{equation}
\underline{\conf}:\body\tto\spc
\end{equation}
will denote the corresponding base map. It is assumed that $\ul{\conf}$
is oriented. Incompatibility occurs when $\kappa$ is not the tangent
mapping $T\underline{\conf}$ of some $\underline{\conf}:\body\to\spc$.
It is assumed that for each $\conf\in\csp$, $\image\conf$ is a subbundle
of the restriction of $T\spc$ to $\image\ul{\conf}$. Thus, for each
$\conf\in\csp$, $\image\conf$ has a structure of a vector bundle
with projection
\begin{equation}
\vbp_{\conf}:\image\conf\tto\image\ul{\conf}.
\end{equation}
(Note that we use the notation $\conf$ for the vector bundle morphism
rather the traditional $F$ in order to emphasize the analogy with
the general case described above.)

Consider, 
\begin{equation}
\underline{\csp}:=\{\underline{\conf}\mid\conf\in\csp\},
\end{equation}
the set of all base mappings corresponding to the vector bundle morphisms
in $\csp$. We have a natural projection
\begin{equation}
B:\csp\tto\ul{\csp}\qquad\conf\lmt\ul{\conf}.
\end{equation}
The compatibility relation $\sim$ is now redefined as follows. The
configurations $\conf_{1}$ and $\conf_{2}$ are compatible, that
is $\conf_{2}\sim\conf_{1}$ if there exists some diffeomorphism,
$\disp_{21}\in\diffs$ of $\spc$, the tangent map, $T\disp_{21}:T\spc\to T\spc$,
of which satisfies
\begin{equation}
\conf_{2}=T\disp_{21}\comp\conf_{1}.
\end{equation}
Clearly, compatibility is an equivalence relation. Note that the collection,
$H$, of mappings $T\spc\to T\spc$ that are of the form $h=T\disp$,
where $\disp$ is a diffeomorphism of $\spc$, is a subgroup of the
group of all diffeomorphisms of $T\spc$.

On $\underline{\csp}$ we can apply the construction described in
Section \ref{sec:The-Basic-Framework}, and define the equivalence
relation 
\begin{equation}
\ul{\conf}_{1}\,\,\underline{\sim}\,\,\,\ul{\conf}_{2}\qquad\text{if}\qquad\conf_{2}=\disp_{21}\comp\conf_{1},
\end{equation}
for some $\disp_{21}\in\diffs$. Evidently, 
\begin{equation}
\ul{\conf}_{1}=B(\conf_{1})\,\,\underline{\sim}\,\,\,\ul{\conf}_{2}=B(\conf_{2}),\qquad\text{if}\qquad\conf_{2}\sim\conf_{1}.\label{eq:comp-goes-to-base-maps}
\end{equation}
The converse is false in general. For two distinct infinitesimal configurations
such that the images of the base mappings are compatible, the infinitesimal
structures need not be compatible. Once again we define the space
of material structures, or the embodiment space, $\mtsp$, to be the
quotient space $\csp/\sim$, and we have the natural projection
\begin{equation}
\pi_{\mtsp}:\csp\tto\mtsp=\csp/\sim.
\end{equation}
By our construction, the vector bundles of the form $\image\conf$
for the various elements $\conf\in\mtn\in\mtsp$ are all vector bundle
diffeomorphic. That is, if $\conf_{2}\sim\conf_{1}$, then,
\begin{equation}
T\disp_{21}\resto{\image\conf_{1}}:\image\conf_{1}\tto\image\conf_{2},
\end{equation}
is a diffeomorphism of vector bundles.

In accordance with the previous section, we write 
\begin{equation}
\ul{\pi}_{\ul{\mtsp}}:\ul{\csp}\tto\ul{\mtsp}:=\ul{\csp}/\ul{\sim},
\end{equation}
for the natural projection induced by the equivalence relation $\ul{\sim}$.

Let $\mtn\in\mtsp$ be represented by $\conf$, and let $\ul{\mtn}=[\ul{\conf}=B(\conf)]\in\ul{\mtsp}$.
It follows from Equation (\ref{eq:comp-goes-to-base-maps}) that $\ul{\mtn}$
is independent of the particular representative $\conf\in\mtn$. Hence,
we have a surjection 
\begin{equation}
B_{\sim}:\mtsp\tto\ul{\mtsp},\qquad[\conf]\lmt[B(\conf)].
\end{equation}
Thus, $B_{\sim}^{-1}(\ul{\mtn})$ is the collection of infinitesimal
material structures for which the base material structure is $\ul{\mtn}$.

\subsection{The structure corresponding to an infinitesimal embodiment}

As in Equation (\ref{eq:disjoint_union}), $A_{\mtn}$, $\mtn\in\mtsp$
is defined as the disjoint union of the images of all $\conf\in\mtn$.
In analogy with (\ref{eq:Eq-Rel-e}), we define the equivalence relation
$\sim_{\mtn}$ in $A_{\mtn}$ by 
\begin{equation}
a_{1}=(v_{1},\conf_{1})\sim_{\mtn}a_{2}=(v_{2},\conf_{2})\qquad\text{if}\qquad v_{2}=T\disp_{21}(v_{1})\label{eq:rel-on-disjoint-union}
\end{equation}
for $\disp_{21}\in\diffs$ satisfying $\conf_{2}=T\disp_{21}\comp\conf_{1}$.
In accordance with our notation scheme, we have
\begin{equation}
\vb_{\mtn}:=A_{\mtn}/\sim_{\mtn},\qquad\pi_{\mtn}:A_{\mtn}\tto\vb_{\mtn}.
\end{equation}

Evidently, if $a_{1}\sim_{\mtn}a_{2}$, as above,
\begin{equation}
\vbp_{\conf_{2}}(v_{2})=\disp_{21}(\vbp_{\conf_{1}}(v_{1})).\label{eq:rel-projected-on-base}
\end{equation}
Let $\ul{\mtn}\in\ul{\mtsp}$, and 
\begin{equation}
\ul A_{\ul e}:=\coprod_{\ul{\conf}\in\ul{\mtn}}\image\ul{\conf}.
\end{equation}
On $\ul A_{\ul{\mtn}}$ we have the equivalence relation
\begin{equation}
\ul a_{1}=(y_{1},\ul{\conf}_{1})\,\,\,\ul{\sim}_{\ul{\,\mtn}}\,\,\,\ul a_{2}=(y_{2},\ul{\conf}_{2})\qquad\text{if}\qquad y_{2}=\disp_{21}(y_{1}),
\end{equation}
for some $\disp_{21}\in\diffs$ satisfying $\ul{\conf}_{2}=\disp_{21}\comp\ul{\conf}_{1}$.
We set 
\begin{equation}
\ul{\body}_{\ul{\mtn}}=\ul A_{\ul e}/\,\ul{\sim}_{\ul{\,\mtn}},\qquad\ul{\pi}_{\ul{\mtn}}:\ul A_{\ul{\mtn}}\tto\ul{\body}_{\ul{\mtn}}.\label{eq:proj_e_base}
\end{equation}
From (\ref{eq:rel-projected-on-base}) it follows that 
\begin{equation}
(v_{1},\conf_{1})\sim_{\mtn}(v_{2},\conf_{2})\qquad\text{implies}\qquad(\vbp_{\conf_{1}}(v_{1}),\ul{\conf}_{1})\,\,\,\ul{\sim}_{\ul{\,\mtn}}\,\,\,(\vbp_{\conf_{2}}(v_{2}),\ul{\conf}_{2}).\label{eq:rel-proj-base-1}
\end{equation}

We consider the quotient space, $\vb_{\mtn}:=A_{\mtn}/\sim_{\mtn}$,
the structure of which is described below. We will show that $\vb_{\mtn}$
is a vector bundle over $\ul{\body}_{\ul{\mtn}}$. The fiber over
$x\in\ul{\body}_{\ul{\mtn}}$ represents the infinitesimal material
structure at $x$.

For $u=[(v,\conf)]_{\mtn}\in\vb_{\mtn}$ ($[\cdot]_{\mtn}$ indicates
the equivalence class relative to $\sim_{\mtn}$), we set
\begin{equation}
\pi_{\vb_{\mtn}}(u):=[(\vbp_{\conf}(v),B(\conf)]_{\ul{\mtn}}=[(\vbp_{\conf}(v),\ul{\conf}]_{\ul{\mtn}}\in\ul{\body}_{\ul{\mtn}}.
\end{equation}
By (\ref{eq:rel-proj-base-1}), $\pi_{\vb_{\mtn}}(u)$ is independent
of the representative $(v,\conf)\in A_{\mtn}$, so we have a projection
\begin{equation}
\pi_{\vb_{\mtn}}:\vb_{\mtn}\tto\ul{\body}_{\ul{\mtn}}.
\end{equation}

Let $\mtn\in\mtsp$ be an embodiment, and let $\conf\in\mtn$. We
have a natural mapping 
\begin{equation}
\pi_{\mtn\conf}:\image\conf\tto\vb_{\mtn},\qquad v\lmt[(v,\conf)]_{\mtn}.\label{eq:pi_ek-1}
\end{equation}
The mapping $\pi_{\mtn\conf}$ is a vector bundle diffeomorphism.

Similarly, let $\ul{\mtn}\in\ul{\mtsp}$, and let $\ul{\conf}\in\ul{\mtn}$.
We have a natural diffeomorphism,
\begin{equation}
\ul{\pi}_{\ul{\mtn}\ul{\conf}}:\image\ul{\conf}\tto\ul{\body}_{\ul{\mtn}},\qquad y\lmt[(y,\ul{\conf})]_{\ul{\mtn}}\label{eq:pi_ek-2}
\end{equation}
as in the previous section.

The induced decomposition is 
\begin{gather}
\ul{\conf}=\ul{\conf}_{\ul e}\comp\ul{\conf}_{a\ul e},\qquad\ul{\conf}_{\ul{\mtn}}:=\incl_{\ul{\conf}}\comp\ul{\pi}_{\ul{\mtn}\ul{\conf}}^{-1},\\
\conf=\conf_{\mtn}\comp\conf_{a\mtn},\qquad\conf_{a\mtn}:=\pi_{\mtn\conf}\comp\conf,\qquad\conf_{e}:=\incl_{\conf}\comp\pi_{\mtn\conf}^{-1}.
\end{gather}
In fact, $\pi_{\vb_{\mtn}}:\vb_{\mtn}\tto\ul{\body}_{\ul{\mtn}}$
is a vector bundle that is the pullback of $\vbp_{\conf}:\image\conf\to\image\ul{\conf}$
by $\ul{\conf}_{\ul{\mtn}}.$ The resulting structure is illustrated
in the following commutative diagram.

\begin{equation}
\xymatrix{ & \text{} &  & \image\conf\subset T\spc\ar[dl]\sp(0.5){\pi_{\mtn\conf}}\ar[ddd]^{\vbp_{\conf}}\\
T\body\ar[rr]_{\conf_{a\mtn}}\ar[rrru]\sp(0.4){\conf}\ar[ddd]_{\tau_{\body}} &  & \vb_{\mtn}\ar@/_{2pc}/[ur]_{\conf_{\mtn}}\ar[ddd]^{\pi_{\vb_{\mtn}}}\\
\\
 &  &  & \image\ul{\conf}\subset\spc\ar[dl]\sp(0.5){\ul{\pi}_{\ul{\mtn}\ul{\conf}}}\\
\body\ar[rr]_{\ul{\conf}_{\ul a\ul{\mtn}}}\ar[rrru]\sp(0.4){\ul{\conf}} &  & \ul{\body}_{\ul{\mtn}}\ar@/_{2pc}/[ur]_{\ul{\conf}_{\ul{\mtn}}}
}
\label{eq:commutative-a}
\end{equation}

In case a system of reference configurations $r:\mtsp\to T\spc$,
a right inverse of $\pi_{\mtsp}$ is given, the comments made in Section
\ref{subsec:Reference-configurations}, still apply. Looking at the
decomposition $\conf=\conf_{\mtn}\comp\conf_{a\mtn}$, we recover
the standard multiplicative decomposition of the deformation gradient
$F$, namely $F=F^{e}F^{p}$.

\section{\label{sec:Dislocated}Deformations of Dislocated Crystals}

The description of dislocated (periodic) crystals falls within the
scheme built up so far. As above, the space manifold, $\spc$, is
an oriented $n$-dimensional manifold and the protobody $T\body$
is the tangent bundle of an oriented, compact, $n$-dimensional manifold
with boundary. A configuration $\conf\in\csp$ is assumed once again
to be a vector bundle morphism
\begin{equation}
\conf:T\body\tto T\spc,
\end{equation}
such that the base mapping, $\ul{\conf}:\body\to\spc$, is an oriented
embedding, and for each $X\in\body$,
\begin{equation}
\conf\resto{T_{X}\body}:T_{X}\body\tto T_{X}\spc
\end{equation}
is an orientation preserving isomorphism.

The tangent bundle, $T\body$, is viewed as the perfect crystal lattice.
Specifically, as a possible interpretation we can say that $\body$
is the set of atoms, itself the lattice structure, while considering
$T\body$ allows us to assign at each point the pertinent optical
axes. See also \cite{Davini1986} where frames at the various material
points represent the crystalline structure.  The fact that $\conf$
need not be $T\ul{\conf}$ reflects the dislocated configuration.

We set $\conf_{1}\sim\conf_{2}$ if there is a diffeomorphism $\disp_{21}:\spc\to\spc$
such that $\conf_{2}=T\disp_{21}\comp\conf_{1}$. Evidently, if $\conf_{1}\sim\conf_{2}=T\disp_{21}\comp\conf_{1}$,
then, $\ul{\conf}_{1}\,\,\ul{\sim}\,\:\ul{\conf}_{2}$. Here, again,
the equivalence relation $\ul{\conf}_{1}\,\,\ul{\sim}\,\:\ul{\conf}_{2}$
is defined by the requirement that there is some diffeomorphism $\disp_{21}$
of $\spc$, such that $\ul{\conf}_{2}=\disp_{21}\comp\ul{\conf}_{1}$.
The spaces $\mtsp$ and $\ul{\mtsp}$ are defined in the previous
section.

Let $\conf_{1},\,\conf_{2}\in\csp$ be arbitrary (not necessarily
related). Then, since both $\ul{\conf}_{1}$ and $\ul{\conf}_{2}$
are embeddings, letting $\ul{\conf}_{1}^{-1}:\image\ul{\conf}_{1}\to\body$
be the right inverse, we have a diffeomorphism 
\begin{equation}
\ul{\conf}_{2}\comp\ul{\conf}_{1}^{-1}:\image\ul{\conf}_{1}\tto\image\ul{\conf}_{2}.
\end{equation}
Since this diffeomorphism may be extended to a diffeomorphism $\disp_{21}$
of $\spc$, all elements $\ul{\conf}\in\ul{\csp}$ are related. This
implies that all $\ul{\conf}\in\ul{\csp}$ share the same embodiment
$\ul{\mtn}=[\ul{\conf}]\in\ul{\mtsp}$, so that $\ul{\mtsp}=\{\ul{\mtn}\}$.

 Moreover, as for any $\conf\in\csp$, $\image\ul{\conf}$ is diffeomorphic
with $\body$ and diffeomorphic with the single $\ul{\body}_{\ul{\mtn}}$
, we may naturally identify $\ul{\body}_{\ul{\mtn}}$ with $\body$
so that $\ul{\conf}_{\ul a\ul{\mtn}}$ is the identity. It follows
that for every $\conf$, $\ul{\conf}_{\ul{\mtn}}=\ul{\conf}$, $\ul{\pi}_{\ul{\mtn}\ul{\conf}}=\ul{\conf}^{-1}:\image\ul{\conf}\to\body$,
and Diagram (\ref{eq:commutative-a}) reduces to
\begin{equation}
\xymatrix{ & \text{} &  & \image\conf\subset T\spc\ar[dl]\sp(0.5){\pi_{\mtn\conf}}\ar[ddd]^{\vbp_{\conf}}\\
T\body\ar[rr]\sp(0.6){\conf_{a\mtn}}\ar[rrru]\sp(0.4){\conf}\ar[dddr]_{\tau_{\body}} &  & \vb_{\mtn}\ar@/_{2pc}/[ur]_{\conf_{\mtn}}\ar[dddl]^{\pi_{\vb_{\mtn}}}\\
\\
 &  &  & \image\ul{\conf}\subset\spc\\
 & \body=\ul{\body}_{\ul{\mtn}}\ar[rru]\sp(0.4){\ul{\conf}}
}
\label{eq:commutative-b}
\end{equation}

As mentioned above, the presence of dislocations, or incompatibility,
is reflected by the fact that $\conf$ is different from $T\ul{\conf}$.
\begin{lem}
Let $\conf_{1},\,\conf_{2}\in\csp$. Then, $\conf_{1}\sim\conf_{2}$
if and only if
\begin{equation}
(T\ul{\conf}_{1})^{-1}\comp\conf_{1}=(T\ul{\conf}_{2})^{-1}\comp\conf_{2},\label{eq:condition1}
\end{equation}
where each side of the equation is a vector bundle morphism $T\body\to T\body$
over the identity, and left inverses of the tangent mappings are well-defined
on the images of the configurations.
\end{lem}

\begin{proof}
Assume that $\conf_{1}\sim\conf_{2}$. Then, there is a diffeomorphism
$\disp_{21}:\spc\to\spc$ such that $\ul{\conf}_{2}=\disp_{21}\comp\ul{\conf}_{1}$
and $\conf_{2}=T\disp_{21}\comp\conf_{1}$. Hence, 
\begin{equation}
\begin{split}(T\ul{\conf}_{2})^{-1}\comp\conf_{2} & =(T(\disp_{21}\comp\ul{\conf}_{1}))^{-1}\comp T\disp_{21}\comp\conf_{1},\\
 & =(T\disp_{21}\comp T\ul{\conf}_{1})^{-1}\comp T\disp_{21}\comp\conf_{1},\\
 & =(T\ul{\conf}_{1})^{-1}\comp\conf_{1},\\
 & =T\ul{\conf}_{1}^{-1}\comp\conf_{1}.
\end{split}
\end{equation}
Conversely, assume that condition (\ref{eq:condition1}) holds. Then,
\begin{equation}
\begin{split}\conf_{2} & =T\ul{\conf}_{2}\comp(T\ul{\conf}_{1})^{-1}\comp\conf_{1},\\
 & =T\ul{\conf}_{2}\comp T\ul{\conf}_{1}^{-1}\comp\conf_{1},\\
 & =T(\ul{\conf}_{2}\comp\ul{\conf}_{1}^{-1})\comp\conf_{1}.
\end{split}
\end{equation}
As mentioned above, $\ul{\conf}_{1}\,\,\ul{\sim}\,\,\ul{\conf}_{2}$
always, and so, there is an extending diffeomorphism $\disp_{21}:\spc\to\spc$
such that $\ul{\conf}_{2}\comp\ul{\conf}_{1}^{-1}$ is the restriction
of $\disp_{21}$ to $\image\ul{\conf}_{1}$ as in the following diagram
\begin{equation}
\xymatrix{ & \text{} &  & T\body\ar[llld]\sb(0.5){\conf_{1}}\ar[dl]\sp(0.5){\conf_{2}}\ar[ddd]^{\tau_{\body}}\\
\image\conf_{1}\ar[rr]_{T(\ul{\conf}_{2}\comp\ul{\conf}_{1}^{-1})}\ar[ddd]_{\tau_{\spc}} &  & \image\conf_{2}\ar[ddd]^{\tau_{\spc}}\\
\\
 &  &  & \body\ar[llld]\sb(0.5){\ul{\conf}_{1}}\ar[dl]\sp(0.5){\ul{\conf}_{2}}\\
\image\ul{\conf}_{1}\ar[rr]_{\ul{\conf}_{2}\comp\ul{\conf}_{1}^{-1}} &  & \image\ul{\conf}_{2}
}
\end{equation}
It follows that $\conf_{2}=T\disp_{21}\comp\conf_{1}$.
\end{proof}
We conclude that for any embodiment $\mtn=[\conf]$, there is a unique
oriented vector bundle isomorphism
\begin{equation}
F_{a\mtn}:T\body\tto T\body,
\end{equation}
over the identity of $\body$. For any $\conf\in\mtn$, $F_{a\mtn}$
satisfies
\begin{equation}
F_{a\mtn}=(T\ul{\conf})^{-1}\comp\conf,
\end{equation}
and this definition is independent of the choice of $\conf$.

Consequently,

\begin{myprop}

For a dislocated crystal, the embodiment space may be identified with
the group of oriented vector isomorphisms $T\body\to T\body$, over
the identity.

\end{myprop}

Any such vector bundle isomorphism may be identified with a section
of a principal fiber bundle over $\body$, the fiber at $X\in\body$
of which is $GL(T_{X}\body)^{+}$. Evidently, under a chart, the fiber
may be modeled by $GL(n)^{+}$, which also acts on the fibers (see
\cite[p. 313]{Ster83}). In fact, we obtain the material $G$-structure
of \cite[p. 261]{Eps10}.
\begin{rem}
As mentioned in the introduction, we view the plastic factor, $\conf_{a\mtn}$,
of the decomposition as the vector bundle morphism that maps the perfect
crystal structure to the dislocated one, an incompatible vector bundle
morphism (as it it not the tangent of the base mapping). The plastic
factor is followed by a compatible (the tangent to the base map) vector
bundle morphism $\conf_{e}:\body_{\mtn}\to\spc$. Their composition
gives the incompatible configuration of the protobody in space. For
this remark, let us refer to this point of view as $II$, and write
\begin{equation}
\conf^{II}=\conf_{\mtn}^{II}\comp\conf_{a\mtn}^{II},
\end{equation}
for a compatible $\conf_{\mtn}^{II}$.

This point of view differs from the point of view (\eg, \cite{Lee1967})
where the body is first dissected into small neighborhoods to release
the residual stresses\textemdash an incompatible mapping\textemdash then
packed into the new configuration in space by another incompatible
mapping, so that the composition is a compatible vector bundle morphism
of the body into space. Let us refer to this point of view as $I$
and write 
\begin{equation}
\conf^{I}=\conf_{\mtn}^{I}\comp\conf_{a\mtn}^{I},
\end{equation}
where now $\conf^{I}$ is compatible.

The relation between the two points of view is quite clear. If we
make the identification
\begin{equation}
\conf_{a\mtn}^{I}=(\conf_{a\mtn}^{II})^{-1}:\body_{\mtn}\tto T\body,\qquad\conf_{\mtn}^{I}=\conf^{II}:T\body\tto T\spc,
\end{equation}
we obtain
\begin{equation}
\conf^{I}=\conf^{II}\comp(\conf_{a\mtn}^{II})^{-1},
\end{equation}
as in the following diagram.
\begin{equation}
\xymatrix{ &  & \image\conf\subset T\spc\\
T\body\ar[rr]^{\conf_{a\mtn}^{II}}\ar[rru]^{\conf^{II}=\conf_{\mtn}^{I}} &  & \body_{\mtn}\ar@/^{1pc}/[ll]^{\conf_{a\mtn}^{I}=(\conf_{a\mtn}^{II})^{-1}}\ar[u]_{\conf^{I}=\conf_{\mtn}^{II}}.
}
\end{equation}
While in point of view $I$, the basic object is the ``frustrated''
body $\body_{\mtn}$, for view $II$, which we adopt in this manuscript,
the basic object is $T\body$, interpreted as the perfect crystal.
\end{rem}

\section{\label{sec:Dislocated-Quasicrystals}Dislocated Quasicrystals}

The above construction admits a natural adaptation to the case of
quasicrystals, i.e., those (natural and synthetic) alloys showing
a quasi-periodic distribution of Bragg's peaks under diffraction experiments.
In fact, every $n$-dimensional quasi-periodic lattice can be considered
as the projection of a periodic atomic array in a $2n$-dimensional
space onto a $n$-dimensional incommensurate subspace. For example,
consider a quasi-periodic lattice in the plane and develop the mass
density function in a Fourier series; quasi-periodicity imposes in
the Fourier series a $4$-dimensional wave vector: once again we go
from n to $2n$ \cite{DM96}.

Quasicrystals admit dislocations \cite{WGZU}, \cite{WaD}. Their
Burgers vector admits a component in the incommensurate subspace and
another one in the orthogonal complement to that space in the higher-dimensional
space from which we construct the quasi-periodic lattice \cite{Kle},
\cite{M19}.

To exploit in this case the structure in previous section, we could
consider $\body$ itself as a locally trivial fiber bundle with base
manifold a fit region in $3D$ real space and $\reals^{3}$ as a typical
fiber. The fit region includes the physical atoms constituting the
body, while the fiber at each point includes information on the low-scale
atomic flips that assure quasi-periodicity in the physical space.
Then we consider $T\body$ and act as above, paying attention to the
circumstance that equivalence relations should account for both basis
and fiber of $\body$ at the same time; in essence they can be considered
as those in the previous section when referred to the higher-dimensional
space from which we obtain the quasi-periodic lattice.

\bigskip{}
\noindent \textbf{\textit{Acknowledgments.}} RS's work is partially
supported by the H.~Greenhill Chair for Theoretical and Applied Mechanics,
and by the Pearlstone Center for Aeronautical Engineering Studies
at Ben-Gurion University of the Negev. \\
This work belongs to activities of the research group ``Theoretical
Mechanics'' in the ``Centro di Ricerca Matematica Ennio De Giorgi''
of the Scuola Normale Superiore in Pisa. PMM acknowledges the support
of GNFM-INDAM.\\
DM acknowledges the support of GNAMPA-INDAM.


\end{document}